\newif\ifarxiv
\arxivtrue

\ifarxiv
\documentclass[pagebackref,11pt]{article}
\usepackage[a4paper,margin=3cm]{geometry}
\else
\documentclass[preprint,12pt,authoryear]{elsarticle}
\fi

\usepackage{microtype}
\usepackage{amssymb}
\usepackage{amsmath}
\usepackage{graphicx} %
\usepackage{mathrsfs}
\usepackage{enumitem}
\usepackage{amsthm}
\usepackage[super]{nth}
\usepackage[toc,page]{appendix}
\usepackage{mathtools}
\usepackage{color}
\usepackage{tabularx}
\usepackage{amsmath}
\usepackage{csquotes}
\usepackage{multicol}
\usepackage[T1]{fontenc}
\usepackage{fancyvrb}
\usepackage{makecell}
\usepackage{dsfont}
\usepackage{caption}

\usepackage{subcaption}
\usepackage{booktabs} 
\usepackage{mdframed}
\usepackage[toc,page]{appendix}
\usepackage{empheq}
\usepackage{arydshln}
\usepackage{nicefrac}
\usepackage{natbib}
\usepackage{hyperref}
\hypersetup{
	breaklinks=true,
	hidelinks,
}
\usepackage{cleveref}
\usepackage{pgfplots}
\pgfplotsset{compat=1.18}
\usepackage[ruled, linesnumbered]{algorithm2e} 

\usepackage{array}
\newcolumntype{L}[1]{>{\raggedright\let\newline\\\arraybackslash\hspace{0pt}}m{#1}}
\newcolumntype{C}[1]{>{\centering\let\newline\\\arraybackslash\hspace{0pt}}m{#1}}
\newcolumntype{R}[1]{>{\raggedleft\let\newline\\\arraybackslash\hspace{0pt}}m{#1}}

\usepackage{tikz}
\usetikzlibrary{automata, positioning, calc, shapes, arrows, fit}

\DeclareMathOperator{\rank}{rank}

\DeclareMathOperator*{\argmax}{arg\,max}

\theoremstyle{plain}
\newtheorem{theorem}{Theorem}

\newtheorem{definition}{Definition}

\newcommand{\GC}{\textnormal{\textsc{greedy capture}}}
\newcommand{\PV}{\textnormal{\textsc{plurality veto}}}

\ifarxiv\else
\usepackage{lineno}
\linenumbers
\usepackage[textsize=scriptsize,backgroundcolor=green!20!white,linecolor=black!60,obeyFinal]{todonotes}
\fi

\ifarxiv
\else
\journal{Operations Research Letters}
\fi

\newcommand{\thetitle}{Proportional Clustering, the $\beta$-Plurality Problem, and Metric Distortion}
\newcommand{\theleon}{Leon Kellerhals}
\newcommand{\thejannik}{Jannik Peters}
\newcommand{\thetuc}{TU Clausthal}
\newcommand{\thetuccountry}{Germany}
\newcommand{\thenus}{National University of Singapore}
\newcommand{\thenuscountry}{Singapore}
\newcommand{\theleonmail}{leon.kellerhals@tu-clausthal.de}
\newcommand{\thejannikmail}{peters@nus.edu.sg}

\newcommand{\theabstract}{%
We show that the proportional clustering problem using the Droop quota for $k = 1$ is equivalent to the $\beta$-plurality problem. We also show that the \PV{} rule can be used to select ($\sqrt{5} - 2$)-plurality points using only ordinal information about the metric space and resolve an open question of \citet{KKK24a} by proving that $(2+\sqrt{5})$-proportionally fair clusterings can be found using purely ordinal information.%
}
\newcommand{\thekeyword}{proportional clustering \sep computational social choice \sep computational geometry \sep voting theory}

\ifarxiv
\title{\thetitle}
\author{\textbf{\theleon}\\{\thetuc, \thetuccountry}\\\texttt{\small\theleonmail} \and \textbf{\thejannik}\\{\thenus}\\\texttt{\small\thejannikmail}}
\date{}
\fi

\begin{document}

\ifarxiv
\maketitle
\begin{abstract}\theabstract\end{abstract}
\else
\begin{frontmatter}
	\title{\thetitle}
	\author[label1]{\theleon}
\affiliation[label1]{organization={\thetuc},%
            country={\thetuccountry}}
\author[label2]{\thejannik}
\affiliation[label2]{organization={\thenus},%
            country={\thenuscountry}}
\begin{abstract}\theabstract\end{abstract}

\end{frontmatter}
\fi

\section{Introduction}
In this research note, we uncover a connection between the proportional clustering \citep{CFL+19a} and the $\beta$-plurality problem \citep{ADGH21a}. The proportional clustering problem is concerned with finding a fairness measure for clustering/facility location problems, ensuring that each sufficiently large set of points to be clustered has a cluster center nearby. If we have $n$ points to cluster and $k$ cluster centers to choose, a natural threshold, or \emph{quota}, for ``sufficiently large'' is $\frac{n}{k}$. In the literature on multiwinner voting, this is historically known as the Hare quota \citep{Tide95a}. This quota serves as the basis for most works on modern multiwinner voting and proportional clustering \citep{ABC+16a, PeSk20a,BrPe23a, MPS23a, KePe23a, ALM23a}. However, it is neither the ``smallest possible'' quota nor the quota most frequently used in real-life elections. Instead, this is the Droop quota $\frac{n}{k+1}$ \citep{Droo81a}, which is also being used (with slight modifications) in real-life elections, for instance in council elections in Scotland \citep{McGr23a}. 

We give a connection between proportional clustering using the Droop quota and the $\beta$-plurality problem
and show that the two are equivalent when fixing the number of cluster centers to $k=1$.
 We additionally uncover a relationship to the metric distortion problem and show that $\beta$-plurality points achieve constant distortion.
 Further, we show that the \PV{} rule of \citet{KiKe22a} (originally designed to achieve good distortion) can be used to find $\Omega(1)$-plurality points, even if the actual distances are unknown and only ordinal preferences are given.
 Our observations from the last result also allow us to show that one can always find a $k$-clustering satisfying $(2+\sqrt{5})$-proportionality using only ordinal information,
 thereby answering an open question of \citet{KKK24a}.

\section{Proportional Clustering and Plurality Points}

In this section, we give a connection between plurality points and proportional clustering with $k=1$ and when using the Droop quota.
We first define the two concepts before proving the connection.     

\subsection{Proportional clustering}

The proportional clustering problem was introduced by \citet{CFL+19a} as a fairness measure for clustering and facility location problems.
They introduced the following notion for the fixed quota $\ell = \frac{n}{k}$. 

\begin{definition}[$\alpha$-approximate $\ell$-proportionality]
	Let $\alpha, \ell \ge 1$.
	For a metric space $(\mathcal X, d)$ and a set $N \subseteq \mathcal X$ of agents, a $k$-clustering $W \subseteq \mathcal X$, $|W| = k$, satisfies \emph{$(\alpha, \ell)$-proportionality} if there is no group $N' \subseteq N$ of at least $|N'| \ge \ell$ agents for which there exists a point $c \in \mathcal X$ such that, for all $i \in N'$,
	\[ \alpha \cdot d(i, c) < d(i, W). \]
	If $\ell=n/k$, then we say the outcome satisfies \emph{$\alpha$-proportionality}.
	If $\ell = \lfloor n/(k+1)\rfloor + 1$, then we say the outcome satisfies \emph{$\alpha$-Droop proportionality}.
\end{definition}
Note that $\ell = \lfloor n/(k+1)\rfloor + 1$ is the smallest integer value such that after removing $k$ blocks of agents each of size $\ell$ there are less than $\ell$ agents left. 

To show that there always exists a clustering satisfying $\alpha$-proportionality, \citet{CFL+19a} introduced the \GC{} algorithm:
Place a ball with radius $y=0$ around each point $p \in \mathcal X$, then smoothly increase $y$.
Once the ball around a point $p$ contains $\ell$ unassigned agents, open a cluster center at $p$ and assign the agents within the ball to $p$.
If the ball around an open center $p$ reaches further unassigned asgents, these are also assigned to $p$.

For $\ell = \frac{n}{k}$ \citet*{CFL+19a} showed that the outcome of \GC{} always satisfies $(1 + \sqrt{2})$-approximate $\ell$-proportionality, while additionally giving instances for any $\varepsilon > 0$ that do not admit clusterings satisfying $(2 - \varepsilon)$-approximate $\ell$-proportionality. These results were improved by \citet{MiSh20a} for Euclidean metric spaces (with the $\ell_2$ norm) for which \GC{} always selects a $2$-proportional outcome, with the corresponding lower bound being $\frac{2}{\sqrt 3}$.

We first note that the upper bounds and proofs of \citet{CFL+19a} and by \citet{MiSh20a} also translate to any integral $\ell > \frac{n}{k+1}$ and thus also to Droop proportionality. 
Indeed, the proofs work for any value of $\ell$ as long as for that $\ell$, \GC{} returns at most $k$ centers.
\begin{theorem}
	Let $\ell \in \mathbb N$ with $\ell > \frac{n}{k+1}$.
	Then \GC{} returns an outcome satisfying $\alpha$-approximate $\ell$-proportionality, where $\alpha \le (1 + \sqrt{2})$ in general metric spaces and $\alpha \le 2$ in Euclidean spaces.
\end{theorem}

Despite the many follow-up works studying several aspects of proportional clustering \citep*{KKK24a, ALM23a, KePe23a, CMS24a}, no work has been able to improve either the lower or upper bounds of \citet*{CFL+19a} and \citet*{MiSh20a} so far.

\subsection{\texorpdfstring{$\beta$}{β}-plurality points}

Independently of the previous discussion, \citet{ADGH21a} introduced the concept of a $\beta$-Plurality point inspired by works on Voronoi games and the spatial theory of voting \citep{Blac48a, Down57a, EnHi83a}.
We are given a population of agents, each corresponding to a point in a metric space, who want to elect another point in the space.
The goal is to find a point such that no other point is ``strongly'' preferred to it by a majority of the agents.
If the goal was to simply find a point such that no other point is preferred by a majority, this would be the classical Condorcet winner problem.
It is commonly known that such a Condorcet winner need not exist, even in $\mathbb{R}^2$, \citep[see e.g.][]{LVV24a}.
\citeauthor{ADGH21a} relax this condition by requiring that no point should be preferred by a majority by at least a factor of $\beta$.\footnote{Here, we state the definition as in \citet[Definition 2.1]{FiFi24a}. The original version \citep{ADGH21a} requires that $\lvert \{i \in N : \beta \cdot d(i,p) < d(i,q) \}\rvert \ge \lvert \{i \in N : \beta \cdot d(i,p) > d(i,q) \}\rvert$. However, existence results for these two versions were shown to be equivalent by \citet{FiFi24a}.}
\begin{definition}
    Let $\beta \le 1$. For a metric space $(\mathcal X, d)$ and a set $N \subseteq \mathcal X$ of agents, a point $p \in \mathcal{X}$ is a $\beta$-plurality point if 
    \[
	    \textstyle\lvert \{i \in N : \beta \cdot d(i,p) \le d(i,q) \}\rvert \ge \frac{\lvert N \rvert}{2} \text{ for all } q \in \mathcal{X}.
    \]
\end{definition}

\citet{ADGH21a} showed that in $\mathcal{X} = \mathbb{R}^2$ the best possible achievable $\beta$ is $\frac{\sqrt{3}}{2}$ and gave an additional lower bound of $\frac{1}{\sqrt{d}}$ for $\mathcal{X} = \mathbb{R}^d$. In a follow-up work, \citet{FiFi24a} showed that for arbitrary metric spaces, the best possible $\beta$ is between $\sqrt{2} - 1$ and $\frac{1}{2}$ and improved the lower bound for $\mathbb{R}^d$ (for $d \ge 4$) to $0.557$.

\subsection{Connecting plurality points and proportionality}

\begin{table}[tb]
\centering
\resizebox{\columnwidth}{!}{%
\begin{tabular}{@{}lcccc@{}} 
\toprule 
  & Arbitrary Metric & $\mathbb{R}^2$  & $\mathbb{R}^d$  \\ %
 \midrule
$\alpha$-proportionality  & $2 \le \alpha \le 1 + \sqrt{2}^\clubsuit$ & $\frac{2}{\sqrt{3}} \le \alpha \le 2^\spadesuit$ & $\frac{2}{\sqrt{3}} \le \alpha \le 2^\spadesuit$ \\ %
$\beta$-Plurality   & $2 \le \frac{1}{\beta} \le 1 + \sqrt{2}^\heartsuit$ & $\frac{1}{\beta} = \frac{2}{\sqrt{3}}^\diamondsuit $ & $\frac{2}{\sqrt{3}}^\diamondsuit \le \frac{1}{\beta} \le \min\left(\sqrt{d}^\diamondsuit, 1.8^\heartsuit \right)$\\ %
\bottomrule
\end{tabular}%
}
    \caption{Comparison of the previously known results. For better comparability, we state the results on $\beta$-plurality in terms of $\frac{1}{\beta}$. Results marked with a $\clubsuit$ are by \citet*{CFL+19a}, with a $\spadesuit$ by \citet{MiSh20a}, $\heartsuit$ by \citet{FiFi24a}, and $\diamondsuit$ by \citet*{ADGH21a}.}
    \label{tab:rules-num-disagreement}
\end{table}
Indeed, if one compares the results achieved for proportional clustering and $\beta$-plurality points, they look very similar, see \Cref{tab:rules-num-disagreement}. This is not a coincidence, as it is easy to show that plurality points are equivalent to proportional clusterings of size one (when using the Droop quota).
\begin{theorem}
    A point $p$ is a $\beta$-plurality point if and only if the clustering $\{p\}$ satisfies $\frac{1}{\beta}$-Droop proportionality.
\end{theorem}
\begin{proof}
    Let $p$ be a $\beta$-plurality point, $N' \subseteq N$ of size $\lfloor\frac{\lvert N\rvert}{2}\rfloor + 1$, and $q \in \mathcal{X}$. Since $p$ is a $\beta$-plurality point, we know that $\lvert \{i \in N : d(i,p) \le \frac{1}{\beta} d(i,q)\}\rvert \ge \frac{\lvert N \rvert}{2}$. Thus, there must exist a point in $N'$ not preferring $q$ by a factor of more than $\frac{1}{\beta}$ and thus $\{p\}$ is $\frac{1}{\beta}$-Droop proportional. Similarly, if $p$ is not a $\beta$-plurality point, the set $N \setminus \{i \in N : d(i,p) \le \frac{1}{\beta} d(i,q)\}$ of more than $\frac{\lvert N \rvert}{2}$ agents can deviate to~$q$; thus $p$ does not satisfy $\frac{1}{\beta}$-Droop proportionality.
\end{proof}

Thus, the stronger framework of proportional clustering provides a generalization of $\beta$-plurality points, and has interestingly enough mostly arrived to the same existential results.\footnote{Sadly, we do not find a way to apply the methods of \citet{FiFi24a} and \citet{ADGH21a} to find stronger bounds for the proportional clustering problem, as they explicitly construct only a single point.}

\section{Plurality Points and Distortion}

The \emph{(metric) distortion} of a point measures how well a given point approximates the ``social cost'', i.e., the total distance to all agents.
When the distance metric is given, finding the minimum distortion point (candidate) is equivalent to the $1$-median problem and thus trivial.
However, when we are only given \emph{ordinal} information, i.e., an ordering of the distances from each point, from closest to furthest, finding the minimum distortion point is more challenging and has gained significant attention in recent years \citep{GHS20a, KiKe22a}.
Formally, distortion is defined as follows.

\begin{definition}
	For a metric space $(\mathcal X, d)$ and a set $N \subseteq \mathcal X$ of agents, the \emph{social cost} of a point $p \in \mathcal{X}$ is 
    $%
    d(N, p) = \sum_{i \in N} d(i,p).
    $ %
    The \emph{distortion} of $p$ is 
    \[
    \sup_{q \in \mathcal{X}}\frac{d(N,p)}{d(N, q)}.
    \]
\end{definition}

It is well known that a $1$-plurality point (i.e., a Condorcet winner) has a distortion of $3$ \citep{ABE+18a}.
We generalize this to the case of arbitrary $\beta$.

\begin{theorem}
    Every $\beta$-plurality point has a distortion of $2\frac 1 \beta + 1$.
\end{theorem}
\begin{proof}
    Let $p$ be a $\beta$-plurality point and $q$ be any other point. Let $N' = \{i \in N : \beta \cdot d(i,p) \le d(i,q)\}$. Then we have 
    \begin{align*}
        d(N, p) &= \sum_{i \in N} d(i,p) = \sum_{i \in N'} d(i,p) + \sum_{i \in N\setminus N'} d(i,p) \\&\le \sum_{i \in N'} d(i,p) + \sum_{i \in N\setminus N'} (d(i,q) + d(q,p)) \\&\le \sum_{i \in N'} (d(i,p)+ d(q,p)) + \sum_{i \in N\setminus N'} d(i,q) \\&\le \sum_{i \in N'} (2d(i,p)+ d(i,q)) + \sum_{i \in N\setminus N'} d(i,q) \le \left(2\frac{1}{\beta} + 1\right) d(N, q).
    \end{align*}
    Here we used the fact that $\lvert N'\rvert \ge \lvert N \setminus N'\rvert$ and the triangle inequality.
\end{proof}

It is, however, easy to see that a constant factor distortion does not imply that a point is a constant plurality point. Consider an instance on the line, with $\frac{n}{2} - 1$ points at $0$, $\frac{n}{2} + 1$ points at $1$. Here, any point at location $0$ has a distortion of at most $2$. However, the $\frac{n}{2} + 1$ agents are a witness that this is not any finite plurality point.

\section{Plurality Points Using Only Ordinal Information}

In their seminal work, \citet{KiKe22a} introduced \PV{} (as a simple variant of the rule of \citet{GHS20a}, see also \citet{KiKe23a}) as a voting rule that uses only ordinal information and always selects a point with distortion $3$.
The rule works in two phases.
In the first phase, each agent nominates their top choice, i.e., each point receives a score equal to its plurality score.
Then the agents go one-by-one and decrement the score of the point with the lowest positive score.
The last point to get decremented wins.
We show that the point selected by \PV{} also is an $\mathcal{O}(1)$-plurality point.

\begin{theorem}
    \PV{} always selects a $(\sqrt{5} - 2)$-plurality point.
\end{theorem}
\begin{proof}
	Let $p$ be the point selected by \PV{} and let $q$ be any other point.
	We choose $\beta' \le 1$ to be the smallest value such that $p$ is not a $\beta'$-plurality point.
	Then $S = \{ i \in N : \beta' d(i, p) > d(i, q) \}$ contains more than $n/2$ points.
	Let $C \subseteq \mathcal X$ be the set of candidates nominated by points in $S$ in the first phase and note that $p \notin C$ by definition of $S$.
	Moreover, the sum of plurality scores of $C$ is greater than $n/2$ (and the score sum of $\mathcal X \setminus C$ is less than $n/2$).
	If all agents in $S$ only decrement scores of candidates outside of $C$, then $p$ will not be selected.
	Thus, some agent in $S$ decrements the score of some $r \in C$.
	Let $j \in S$ be the agent that nominated $r$.
	Then we have
	$d(i, p) \le d(i, r) \le d(i, q) + d(q, j) + d(j, r) \le d(i, q) + 2d(j, q)$.
	Similarly,
	$d(j, p) \le d(j, q) + d(q, i) + d(i, p)$.
	By definition of $S$, $\frac{1}{\beta'}$ is less than
	\begin{align*}
		\min\left(\frac{d(i,p)}{d(i,q)}, \frac{d(j,p)}{d(j,q)}\right) &\le \min\left(\frac{d(i,q) + 2d(j,q)}{d(i,q)}, \frac{d(j,q) + d(i,q) + d(i,p)}{d(j,q)}\right) \\& \le \min\left(\frac{d(i,q) + 2d(j,q)}{d(i,q)}, \frac{3d(j,q) + 2d(i,q)}{d(j,q)}\right)
											    \\ & \le \max_{x \ge 0} \min\left(1 + 2x, 3 + \frac{2}{x}\right) = 2 + \sqrt{5} = \frac{1}{\sqrt{5}-2}.
	\end{align*}
	Therefore, \PV{} always selects a $(\sqrt{5} - 2)$-plurality point.
\end{proof}

\section{Proportional Clusterings Using Only Ordinal Information}

We next turn to the problem of finding proportional clusterings when we only have ordinal information at hand.
This problem was considered by \citet{KKK24a} who showed that one can always find a $\frac{5 + \sqrt{41}}{2}$-proportional clustering, but there are instances where no (purely ordinal) algorithm can satisfy $\alpha$-proportionality for $\alpha < 2 + \sqrt{5}$.
For their upper-bound, they employed the Expanding Approvals Rule (EAR) by \citet{AzLe20a}.

We close the gap by showing that one can always find a $(2+\sqrt{5})$-proportional clustering using only ordinal information.
Indeed, to prove this, it suffices to use the property of \emph{rank-JR}, which was conceptually introduced by \citet{BrPe23a}.
We again define an $\ell$-quota version of the notion.
To this end, we say that an agent $i$ has $\rank(i, c)$ for candidate $c$ if this candidate is the $\rank(i, c)$-th closest candidate to this agent.
An outcome $W$ of size $k$ satisfies $\ell$-rank-JR,
if for every rank $r$, every set $N' \subseteq N$ of size at least $\ell$ such that there is a candidate $c$ with $\rank(i, c) \le r$ for each $j \in N'$,
there is at least one $w \in W$ and $i \in N'$ such that $\rank(i, w) \le r$.
We remark that an outcome returned by EAR satisfies $\ell$-rank-JR with $\ell \ge \frac{n}{k}$.

\begin{theorem}
	Let $\ell \in \mathbb{N}$ with $\ell > \frac{n}{k+1}$ and $W$ be an outcome satisfying $\ell$-rank-JR.
	Then $W$ satisfies $(2+\sqrt{5})$-approximate $\ell$-proportionality.
\end{theorem}
\begin{proof}
	Let $c \notin W$ be an unchosen candidate and $N' \subseteq N$ be a subset of agents of size at least $\ell$.
	Let $i \in N'$ be the agent that gives the greatest rank to $c$, let $r = \rank(i, c)$, and let $C_r$ be the set of the $r$ candidates that are closest to $i$.
	Note that $d(i, c_r) \le d(i, c)$ for any $c_r \in C_r$.
	As $W$ satisfies $\ell$-rank-JR, there is an agent $j \in N$ and a candidate $w \in W$ such that $\rank(j, w) \le r$.
	Then, for each $c_r \in C_r$ we have $d(j, c_r) \le d(j, c) + d(c, i) + d(i, c_r) \le d(j, c) + 2d(i, c) \eqqcolon y$.
	As there are $r$ candidates in $C_r$, the distance of $j$ to its $r$-th ranked candidate is at most $y$; thus also $d(j, w) \le y$.
	Putting this together, we obtain that $W$ satisfies $\alpha$-proportionality, where $\alpha$ is at most
	\begin{align*}
		\min\left(\frac{d(i,w)}{d(i,c)}, \frac{d(j,w)}{d(j,c)}\right)
		&\le \min\left(\frac{d(i,c) + d(j,c) + d(j,w)}{d(i,c)}, \frac{d(j,w)}{d(j,c)}\right)\\
		&\le \min\left(\frac{3d(i,c) + 2d(j,c)}{d(i,c)}, \frac{d(j, c) + 2d(i,c)}{d(j,c)}\right)\\
		&\le \max_{x \ge 0} \min\left(3 + \frac{2}{x}, 1 + 2x\right) = 2 + \sqrt{5}.\qedhere
	\end{align*}
\end{proof}

In fact we can also show that rank-PJR (a slightly stricter notion than rank-JR) implies an approximation to the core of \citet{EbMi23a}.
As this result is mostly a modification of the previous theorem as well as the proofs of \citet{KePe23a} we defer definitions and proof to the appendix.
\begin{theorem}
    Let $W$ be a committee satisfying rank-PJR. Then $W$ is in the $(4 + \sqrt{13})$-$q$-core for all $q \le k$.
\end{theorem}
\section{Conclusion}
We observed a simple equivalence between the Droop proportional clustering problem and the $\beta$-plurality problem. At the moment, for arbitrary metric spaces, both problems are stuck at the exact same bounds. Is it perhaps possible to first ``attack'' the simpler $\beta$-plurality problem, e.g., can maybe more sophisticated methods lead to improved bounds for $k = 1$ and can these be translated to larger $k$? Relatedly, can the methods of \citet{ADGH21a} and \citet{FiFi24a} for Euclidean spaces be generalized to the clustering problem? Finally, \citet{MiSh20a} showed that for the $\ell_1$ and $\ell_\infty$ norms \GC{} does not improve upon the $1 + \sqrt 2$ bound. Is it still possible to obtain better bounds in these spaces?

\section{Acknowledgements}
This research was funded by the Singapore Ministry of Education under grant
number MOE-T2EP20221-0001. We thank Warut Suksompong for helpful comments.
\bibliographystyle{abbrvnat}
\bibliography{abb, algo}

\appendix 

\section{A connection between rank-PJR and the core}

We first define the \emph{$q$-core} notion---a generalization of proportionality in which an agent is not represented by the closest center but the $q$ closest centers.
For $W \subseteq \mathcal X$ and $q \le \lvert W \rvert$, define $d^q(i, W)$ be the distance of $i$ to the $q$-th clostest point in $W$.
By the triangle inequality, $d^q(i, W) = d(i, j) + d^q(j, W)$, for $i, j \in \mathcal X$.

As with the previous notions, we can also lift the core notion to any quota $\ell > \frac{n}{k+1}$.
\begin{definition}
    \label{def:core}
    For $\alpha, \ell \ge 1$ an outcome $W$ is in the \emph{$\alpha$-approximate $\ell$-quota $q$-core}, if there is no $\mu \in \mathbb{N}$ and no $N' \subseteq N$ with $|N'| \ge \mu \cdot \ell$ and set $C' \subseteq \mathcal X$ with $q \le \lvert C' \rvert \le \mu$ such that $\alpha \cdot d^q(i, C') < d^q(i, W)$ for all $i \in N'$.
\end{definition}

Let us further define $\ell$-rank-PJR, which is a restriction of $\ell$-rank-JR.
An outcome $W$ of size $k$ satisfies $\ell$-rank-PJR,
if for every rank $r$, every $\mu \ge 1$,
and every set $N' \subseteq N$ of size at least $\mu \cdot \ell$ such that there is a set $C'$ of at least $\mu$ candidates with $\rank(i, c) \le r$ for each $i \in N'$ and $c \in C'$,
there is a set $W' \subseteq W$ of at least $\mu$ winners such that for each $w \in W'$ there is an agent $i' \in N'$ with $\rank(i', w') \le r$.

\begin{theorem}
	Let $\ell \in \mathbb N$, $\ell > \frac{n}{k+1}$, and let $W$ be an outcome satisfying $\ell$-rank-PJR.
	Then $W$ is in the $(4 + \sqrt{13})$-approximate $\ell$-quota $q$-core for every $q \le k$.
\end{theorem}
\begin{proof}
	Let $N' \subseteq N$ be a subset of at least $\mu \cdot \ell$ agents and let $C' \subseteq \mathcal X$ with $q \le C' \le \mu$.
	Assume that each agent in $N'$ marks their $q$ favorite candidates that are in $C'$.
	This places $|N'| \cdot q \ge \mu \cdot \ell \cdot q \ge |C'| \cdot \ell \cdot q$ marks on the candidates in $C'$; thus there exists a candidate $c \in C'$ with at least $\ell \cdot q$ marks, and each mark is by a unique agent in $N'$.
	Let $N''$ be the set of at least $\ell \cdot q$ agents that marked $c$ and note that each $j \in N''$ has $c$ among their $q$ favorite candidates among $C'$, that is, $d(j, c) \le d^q(j, C')$.

	Let $i \in N''$ be the agent maximizing $d(i, c)$ and let $C_i^q$ be the set of $i$'s $q$ favorite candidates in $C'$.
	Let $j = \argmax_{j \in N''} \max_{c' \in C_i^q} \rank(j, c')$ be the agent assigning the maximum rank to any candidate in $C_i^q$ and let $r = \max_{c' \in C_i^q} \rank(j, c')$ be the corresponding rank.
	Then every of the at least $q \cdot \ell$ agents in $N''$ must assign rank at most $r$ to all candidates in $C_i^q$.
	As $W$ satisfies $\ell$-rank-PJR, there are at least $q$ candidates $W_i^q \subseteq W$ such that for each $w \in W_i^q$ there is an agent $i' \in N''$ with $\rank(i', w) \le r$.

	For any agent $h \in N''$ let $c_h^r$ be their candidate at rank $r$.
	Then
	\begin{align*}
		d(h, c_h^r) &\le d(h, j) + \max_{c' \in C_i^q} d(j, c')\\
			    &\le d(h, c) + d(c, j) + d(j, c) + d(i, c) + \max_{c' \in C_i^q} d(i, c')\\
			    &\le d^q(i, C') + 2d^q(j, C') + 2d^q(i, C') = 3d^q(i, C') + 2d^q(j, C').
	\end{align*}
	Moreover, using that $d(i, h) \le d(i, c) + d(c, h) \le 2d^q(i, C')$,
	we obtain
	\[
		d^q(i, W) \le \max_{h \in N''} \left( d(i, h) + d(h, c_h^r) \right) \le 5d^q(i, C') + 2d^q(j, C').
	\]
	Lastly, as $d(j, i) \le d(j, c) + d(c, i) \le d^q(j, C') + d^q(i, C')$, we have
	\[
		d^q(j, W) \le d(j, i) + d^q(i, W) \le 6d^q(i, C') + 3d^q(j, C').
	\]
	Putting this together, $W$ is in the $\alpha$-approximate $\ell$-quota $q$-core where
	\begin{align*}
		\alpha &\le \min \left( \frac{d^q(i, W)}{d^q(i, C')}, \frac{d^q(j, W)}{d^q(j, C')} \right)\\ %
		&\le \min \left( \frac{5d^q(i, C') + 2d^q(j, C')}{d^q(i, C')}, \frac{6d^q(i, C') + 3d^q(j, C')}{d^q(j, C')} \right)\\
		&\le \max_{x \ge 0} \min(5 + \frac{2}{x}, 6x + 3) = 4 + \sqrt{13}.\qedhere
	\end{align*}
\end{proof}
\end{document}